\documentclass[journal,final,twocolumn]{IEEEtranTCOM}

\usepackage{graphicx,cite,epsfig,amssymb,amsmath,url,stfloats,latexsym}
\usepackage{array}
\usepackage{arydshln}
\usepackage{amsfonts}
\usepackage{pgfplots}
\usepackage{algorithm}
\usepackage{algorithmic}
\usepackage{booktabs} % For better table formatting
\usepackage{epstopdf}
\usepackage{amsfonts,amsthm}
\usepackage{multirow}
\usepackage{mathrsfs}
\usepackage{xcolor}
\usepackage{amsfonts}
\usepackage{graphicx}
\usepackage{tabularx}
\usepackage{array}
\usepackage{arydshln}
\usepackage{multicol}
\usepackage{multirow} 
\usepackage{mathtools}
\usepackage[font=scriptsize]{subfig}
\usepackage[font=small]{caption}
\usepackage{textcomp}
\usepackage{stfloats}
\usepackage{booktabs}
\usepackage{verbatim}
\usepackage{graphicx}
\usepackage{cite}
\usepackage{url}

\let\labelindent\relax  %Since the \labelindent command exists for legacy reasons in the IEEE template, you can simply "disable" it by adding the following before importing the enumitem package:
\usepackage{enumitem}

\usepackage{amssymb}
\usepackage{amsthm,extpfeil}

\definecolor{mygray}{gray}{.9}
\graphicspath{{figures/}}
\newcolumntype{C}[1]{>{\PreserveBackslash\centering}p{#1}}
\newcolumntype{R}[1]{>{\PreserveBackslash\raggedleft}p{#1}}
\newcolumntype{L}[1]{>{\PreserveBackslash\raggedright}p{#1}}

\usepackage[colorlinks=false,linkcolor=black, bookmarksnumbered]{hyperref}

\newtheorem{theorem}{Theorem}

\newtheorem{example}{Example}

\usepackage[flushleft]{threeparttable}

\newcommand\CPM{\text{CPM}}
\newcommand\PaG{\text{PaG}}

\begin{document}

\sloppy
\title{A Global Coding Scheme for OFDM over Finite Fields}
\author{Juane Li, Qi-yue Yu,~\IEEEmembership{Senior Member,~IEEE}, Khaled Abdel-Ghaffar, Shu Lin, ~\IEEEmembership{Life~Fellow,~IEEE}
%\thanks{Juane Li is with Micron Technology Inc., San Jose, CA 95131, USA (e-mail:jueli@ucdavis.edu). Xin Xiao is with the Department of Electrical and Computer Engineering, University of Arizona, Tucson, AZ 85712, USA (e-mail: 7xinxiao7@email.arizona.edu). Shu Lin and Khaled Abdel-Ghaffar are with the Department of Electrical and Computer Engineering, University of California, Davis, CA 95616, USA (e-mails: \{shulin, kaghaffar\}@ucdavis.edu).}
} 
\maketitle

\begin{abstract}
This paper proposes a highly efficient global coded-multiplexing scheme, conceptualized as Orthogonal Frequency Division Multiplexing over a finite field (FF-OFDM), for reliable multiuser communications. By utilizing a prime length cyclic code and its Hadamard equivalents as algebraic subcarriers, independent data streams are globally multiplexed via a Galois Fourier Transform (GFT) without rate loss. We show that this finite-field synthesis intrinsically generates a global Quasi-Cyclic Low-Density Parity-Check (QC-LDPC) code over $\mathrm{GF}(2^s)$, whose parity-check matrix is governed by the structural rigor of partial geometries. At the receiver, supported by a binary decomposition theorem, the received nonbinary global codeword is jointly decoded using parallel binary iterative soft-decision algorithms prior to demultiplexing. This joint decoding enables seamless reliability information sharing across all user streams, achieving near-bound error performance, rapid convergence without error floors, and strictly linear amortized decoding complexity.
\end{abstract}


\section{Introduction}\label{sect1}

Multiplexing techniques are fundamental to the architecture of high-speed communication systems, enabling the simultaneous transmission of multiple information streams over a shared medium. Historically, this domain has evolved from basic time division multiplexing (TDM) and frequency division multiplexing (FDM) to the highly spectral-efficient orthogonal frequency division multiplexing (OFDM) \cite{OFDM1, OFDM2}. In recent years, the field has seen further innovations such as random multiplexing \cite{RandomMul}. By explicitly decoupling the multiplexing operations from the physical channel characteristics, this approach extends its applicability to arbitrary norm-bounded and spectrally convergent channel matrices.

Among these techniques, OFDM serves as the quintessential example of \textit{multiplexing in the classical signal domain}. Its core principle relies on \textit{frequency-domain multiplexing}: multiple user data streams are modulated onto orthogonal subcarriers and then synthesized into a single composite signal. This allows the receiver to separate the streams efficiently, provided that orthogonality is maintained. Essentially, OFDM transforms the problem of transmitting multiple streams into a parallel transmission over the continuous frequency spectrum.

In a parallel line of research, the \textit{finite-field multiple-access (FFMA) technique} was recently proposed \cite{FFMA1, FFMA_ITW_2024, FFMA2, FFMA3, FFMA_VTC_2025, FFMA_GC_2025, FFMA_polar, FFMA_5m}. Conventional multiple access (MA) techniques generally perform channel coding first, followed by multiplexing, which inherently treats coding and multiplexing as two distinct and decoupled modules \cite{Ref_1, Ref_2, Ref_3, Ref_4, Ref_5, Ref_6, Ref_7, Ref_8}. In contrast, FFMA deliberately reverses this conventional order. By multiplexing the data before applying channel coding, the FFMA system effectively fuses these traditionally separate modules into a unified \textit{coded-multiplexing} framework, referred to as FFMA-coding. This structural integration efficiently addresses the multiuser finite blocklength (FBL) challenge, allowing the receiver to jointly decode the entire FFMA-coded block before demultiplexing the individual user streams.

Furthermore, rather than relying on orthogonal physical resources (such as time or frequency slots) to separate users, FFMA distinguishes them algebraically by assigning unique element-pairs (EPs) to different users or bits. Constructed over finite fields, these EPs serve as virtual resources. The Cartesian product of $M$ distinct EPs forms an EP code, which must satisfy the \textit{unique sum-pattern mapping (USPM)} structural constraint to unambiguously separate users within the finite field. Recent studies \cite{FFMA2} have also proven that any $(n, k)$ linear channel code can function as a codeword-wise EP code, since its $k$ linearly independent codewords act as a basis in the finite field. In other words, a codeword itself can be fundamentally regarded as an algebraic virtual resource to multiplex data streams.

Despite these theoretical advances, a fundamental bottleneck remains in scaling such integrated systems for powerful high-capacity communications. In conventional multi-user designs, multiplexing merely distributes data to avoid interference, contributing little to the system's error-correction capability. Furthermore, while algebraic multiplexing bridges the gap between coding and multiplexing, directly scaling these coupled schemes to high-order Galois fields $\mathrm{GF}(2^s)$ inevitably encounters the prohibitive computational complexity of nonbinary joint decoding algorithms. This complexity bottleneck severely limits the achievable block-length diversity in massive data stream transmission.

To fundamentally overcome this bottleneck, inspired by the algebraic virtual resource concept and the structural elegance of classical OFDM, we propose a paradigm-shifting global coded-multiplexing scheme conceptualized as \textit{Orthogonal Frequency Division Multiplexing over a finite field (FF-OFDM)}.  This scheme is specifically tailored for high-capacity point-to-point and broadcast communication scenarios. Just as OFDM utilizes orthogonal frequencies to multiplex signals in the complex domain, our proposed model utilizes \textit{Hadamard equivalent codes} \cite{Ref_9, Ref_10, Ref_11} to multiplex information streams in the finite field domain. In this framework, the ``subcarriers'' are orthogonal-like algebraic codes acting as virtual resources, and the synthesis of the ``multiplexed signal'' corresponds to the global coupling of these codewords. This approach successfully translates the benefits of signal-domain orthogonality into the error-correction domain, fusing multiplexing and cooperative coding into a single unified operation.

In the proposed scheme, a cyclic code $C$ of prime length $n$ acts as the \emph{base code}. This code, along with its $n - 2$ Hadamard equivalent codes \cite{Ref_11}, form the basis for our finite-field multiplexing. Together with an $(n, n - 1)$ single parity-check (SPC) code, they strictly support the transmission of $ns$ distinct information streams within a unified framework, where $n$ is a prime factor of $2^s - 1$ and $s \geq 3$.

At the multiplexing stage, the $ns$ codewords from the input streams are globally coupled into a single global codeword over $\mathrm{GF}(2^s)$. This global codeword resides in the null space of a low-density parity-check (LDPC) matrix derived from the line-point incidence matrix of a subgeometry of a partial geometry. The collection of these globally coupled codewords forms a structured quasi-cyclic (QC) LDPC code over $\mathrm{GF}(2^s)$.

Crucially, to circumvent the nonbinary decoding bottleneck, the FF-OFDM global code is intrinsically designed to be decomposed and decoded iteratively using a \textit{binary soft-decision decoding algorithm} \cite{Ref_11, Ref_12}. The process of mapping messages into codewords and algebraically combining them is referred to as \textit{global encoding}, which effectively functions as the finite-field multiplexer. At the receiving end, demultiplexing and decoding are not isolated operations; rather, they are performed jointly. The received global codeword is first decoded using a parallel binary iterative soft-decision decoding algorithm, such as the sum-product algorithm (SPA) \cite{Ref_13}, the min-sum algorithm (MSA) \cite{Ref_14}, or their variations \cite{Ref_11, Ref_12}. This joint decoding paradigm transforms independent data streams into a mutually protective global graph, allowing reliability information to be shared seamlessly among all multiplexed streams prior to demultiplexing, which significantly enhances the error performance of each individual stream.

In summary, the proposed FF-OFDM scheme is a profound integration of Hadamard equivalent codes, block interleaving, Galois Fourier transform, and partial geometries. It functions as a powerful \textit{coded-multiplexing technique}, where independent data streams are synthesized without rate loss into an ultra-long, highly efficient LDPC code. This ensures near-capacity reliable transmission, harvesting massive cooperative coding gains with strictly manageable linear decoding complexity.

The rest of this paper is organized as follows. Section \ref{sect2} establishes the mathematical preliminaries, detailing the algebraic properties of cyclic codes, Hadamard equivalents, and partial geometries. Section \ref{sect3} introduces the architecture of the proposed FF-OFDM system, delineating the GFT-based multiplexing transmitter and the joint parallel receiver. In Section \ref{sect4}, we provide a rigorous theoretical analysis of the globally coupled QC-LDPC code and formalize the Binary Decomposition Theorem. Section \ref{sect5} defines the error performance metrics and evaluates the amortized decoding complexity. Section \ref{sect6} presents extensive simulation results across high-capacity point-to-point and broadcast scenarios to validate the proposed scheme. Finally, Section \ref{sect7} concludes the paper with some remarks.

\section{Preliminaries and System Components}\label{sect2}
To establish the FF-OFDM framework, we first formalize the algebraic ``subcarriers'' and the finite-field transform matrices, linking them to the geometric structures that guarantee excellent decoding performance.

\subsection{Cyclic Codes of Prime Lengths and Their Hadamard Equivalents}

For $s \geq 3$, let $\mathrm{GF}(2^s)$ be the Galois field of characteristic 2 with $2^s$ elements. Let $n$ be a prime factor of $2^s - 1$ and $\beta$ be an element of order $n$ in the field $\mathrm{GF}(2^s)$. The elements $\beta^0 = 1, \beta, \beta^2, \ldots, \beta^{n - 1}$ in $\mathrm{GF}(2^s)$ form a cyclic subgroup $\bf G$ of $\mathrm{GF}(2^s)$ with order $n$. For $1 \leq m < n$, let $C$ be an $(n, n - m)$ cyclic code over $\mathrm{GF}(2)$ of length $n$ with dimension $n - m$ whose generator polynomial ${\bf g}(X)$ has $\beta^{l_0}, \beta^{l_1}, \ldots, \beta^{l_{m - 1}}$ as roots, where $0 \leq l_0, l_1, \ldots, l_{m-1} < n$. The generator polynomial 
\begin{equation}
  {\bf g}(X) = (X + \beta^{l_0})(X + \beta^{l_1}) \cdots (X + \beta^{l_{m - 1}}) 
\end{equation}
of the cyclic code $C$ is a polynomial over $\mathrm{GF}(2)$ of degree $m$ \cite{Ref_9, Ref_11}. 

The parity-check matrix of the cyclic code $C$ in terms of the roots of its generator polynomial is an $m \times n$ matrix over $\mathrm{GF}(2^s)$, given by:
\begin{equation}\label{eq:base_matrix_cyclic_code}
\begin{array}{ll}
  {\bf B} &= [\beta^{j l_i}]_{0 \leq i < m, 0 \leq j < n}  \\
  &= \begin{pmatrix}
    1 & \beta^{l_0} & \beta^{2l_0} & \cdots & \beta^{(n - 1)l_0} \\
    1 & \beta^{l_1} & \beta^{2l_1} & \cdots & \beta^{(n - 1)l_1} \\
    \vdots & \vdots & \vdots & \ddots & \vdots \\
    1 & \beta^{l_{m-1}} & \beta^{2l_{m-1}} & \cdots & \beta^{(n - 1)l_{m-1}} 
  \end{pmatrix}. 
\end{array}
\end{equation}

The null space over $\mathrm{GF}(2)$ of $\bf B$ gives the cyclic code $C$. An $n$-tuple $\mathbf{v} = (v_0, v_1, \ldots, v_{n-1})$ over $\mathrm{GF}(2)$ is a codeword in $C$ if and only if $\mathbf{v} \cdot {\bf B}^{\rm T} = \mathbf{0}$, where ${\bf B}^{\rm T}$ is the transpose of $\bf B$. In polynomial form, the vector $\mathbf{v}$ is represented by a polynomial ${\bf v}(X) = v_0 + v_1 X + \cdots + v_{n-1} X^{n - 1}$ over $\mathrm{GF}(2)$ of degree $n - 1$ or less. A polynomial ${\bf v}(X)$ is a code polynomial in $C$ if and only if it is divisible by the generator polynomial ${\bf g}(X)$. 

For $1 \leq k < n$, raising every entry in $\bf B$ to a power of $k$, we obtain the following matrix over $\mathrm{GF}(2^s)$:
\begin{equation}\label{eq:hadamard_power_matrix}
\begin{array}{ll}
  {\bf B}^{\circ k} & = [\beta^{k j l_i}]_{0 \leq i < m, 0 \leq j < n}  \\ 
  & = \begin{pmatrix}
1 & \beta^{k l_0} & \beta^{2k l_0} & \cdots & \beta^{(n - 1)k l_0} \\
1 & \beta^{k l_1} & \beta^{2k l_1} & \cdots & \beta^{(n - 1)k l_1} \\
\vdots & \vdots & \vdots & \ddots & \vdots \\
1 & \beta^{k l_{m-1}} & \beta^{2k l_{m-1}} & \cdots & \beta^{(n - 1)k l_{m-1}} 
\end{pmatrix}.
\end{array}
\end{equation}
This matrix is called the $k$-th \textit{Hadamard power} of $\bf B$ \cite{Ref_11}. Label the columns of both ${\bf B}$ and ${\bf B}^{\circ k}$ from $0$ to $n - 1$. For $1 \leq t < n$, let $j$ be the remainder resulting from dividing $tk$ by $n$. Since $n$ is a prime, the mapping $t \leftrightarrow j$ is a one-to-one mapping for every $k$ ($1 \leq k < n$). Consequently, the $t$-th column of ${\bf B}^{\circ k}$ is identical to the $j$-th column of ${\bf B}$.
The $k$-th Hadmard power of ${\bf B}$ simply permutes the columns of ${\bf B}$.
Let $\pi_k$ denote this column permutation, called the $k$-th Hadamard permutation. Then, ${\bf B}^{\circ k} = \pi_k({\bf B})$.

The null space over $\mathrm{GF}(2)$ of ${\bf B}^{\circ k}$ gives an $(n, n - m)$ code $C^{\circ k}$, which is an \textit{equivalent code} of $C$ \cite{Ref_11}. If $\mathbf{v}$ is a codeword in $C$, the $k$-th Hadamard permutation $\pi_k(\mathbf{v})$ gives a codeword in $C^{\circ k}$. The code $C^{\circ k}$ is called the $k$-th Hadamard equivalent of $C$, with its generator polynomial having $\beta^{kl_0}, \beta^{kl_1}, \ldots, \beta^{kl_{m-1}}$ as roots.

For $k = 1$, ${\bf B}^{\circ 1} = {\bf B}$ and ${C}^{\circ 1} = {C}$. For $k = 0$, ${\bf B}^{\circ 0}$ is an all-ones matrix whose null space over $\mathrm{GF}(2)$ gives an $(n, n - 1)$ single parity-check (SPC) code ${C}_0 = {C}^{\circ 0}$. Hence, there are $n$ distinct codes (the SPC code and $n-1$ Hadamard equivalents) that form the mathematical basis for the FF-OFDM ``subcarriers.''

\subsection{Vandermonde Matrix and Partial Geometry}
To facilitate the global coupling of the subcarriers, we construct a Vandermonde matrix based on the cyclic subgroup $\bf G$ \cite{Ref_11}. This matrix encapsulates the parity-check matrix $\bf B$ as a submatrix. Let $\beta$ be the generator element of $\bf G$. Using the $n$ elements in $\bf G$, we form the following $n \times n$ Vandermonde matrix:            
\begin{equation}\label{eq:vandermonde_matrix}
\arraycolsep=2pt\def\arraystretch{1}
  {\bf V} = [\beta^{ij}]_{0 \leq i, j < n}  = \begin{pmatrix}
 1 & 1 & 1 & \cdots & 1 \\
 1 & \beta & \beta^2 & \cdots & \beta^{n - 1} \\
 1 & \beta^2 & (\beta^2)^2 & \cdots & (\beta^2)^{n - 1} \\
 \vdots & \vdots & \vdots & \ddots & \vdots \\
 1 & \beta^{n - 1} & (\beta^{n - 1})^2 & \cdots & (\beta^{n - 1})^{n - 1}
 \end{pmatrix}.
\end{equation}
The Vandermonde matrix ${\bf V}$ is nonsingular, and its inverse is given by ${\bf V}^{-1} = [\beta^{-ij}]_{0 \le i, j < n}$.

For $0 \leq j < n$, we represent the element $\beta^j$ by a circulant permutation matrix (CPM), denoted by $\CPM(\beta^j)$, over $\mathrm{GF}(2)$ of size $n \times n$. Its top row ${\bf w}$ has a ``1'' at position $j$ and ``0'' elsewhere. Dispersing each entry $\beta^{ij}$ in ${\bf V}$ into an $n \times n$ CPM, we obtain an array of $n \times n$ CPMs:
\begin{equation}\label{eqn:CPM_V}
  \CPM({\bf V}) = [\CPM(\beta^{ij})]_{0 \leq i,j < n}.
\end{equation}               
This forms an $n^2 \times n^2$ matrix over $\mathrm{GF}(2)$ with both column and row weights $n$. 

It has been proven that the matrix CPM($\bf V$) is the line-point incidence matrix of a partial geometry, denoted by PaG$(n, n, n-1)$, which consists of $n^2$ points and $n^2$ lines \cite{Ref_11}. The partial geometry PaG$(n, n, n-1)$ has the following fundamental structural properties: (1) Each line consists of $n$ points; (2) Any two points are on at most one line; (3) Each point is intersected by $n$ lines;
(4) If a point $\bf p$ is not on a line $L$, then there are exactly $n - 1$ lines, each passing through $\bf p$ and a point on $L$;
(5) Two lines either are parallel or intersect at one and only one point; and (6) Each line has $n-1$ line parallel to it. The fundamental structural properties of PaG$(n, n, n-1)$ dictates that any two rows or columns in CPM($\bf V$) share at most one 1-entry, referred to as the \textit{row-column (RC)} constraint \cite{Ref_11, Ref_12}.

Since ${\bf B}$ is an $m \times n$ submatrix of ${\bf V}$, its CPM-dispersion 
\begin{equation} \label{eqn:CPM_B}
  \CPM({\bf B}) = [\CPM(\beta^{jl_i})]_{0 \le i < m, 0 \le j < n}
\end{equation}
is an $mn \times n^2$ matrix over $\mathrm{GF}(2)$. Hence, $\CPM({\bf B})$ is the incidence matrix of a subgeometry of $\PaG(n, n, n - 1)$, strictly satisfying the RC-constraint. This structural trait ensures that the Tanner graph associated with $\CPM({\bf B})$ has a girth of at least 6 \cite{Ref_11, Ref_18}, serving as an ideal foundation for a high-performance Quasi-Cyclic (QC) LDPC code.

\begin{figure*}[tbp]
  \centering
  \includegraphics[width=0.7\textwidth]{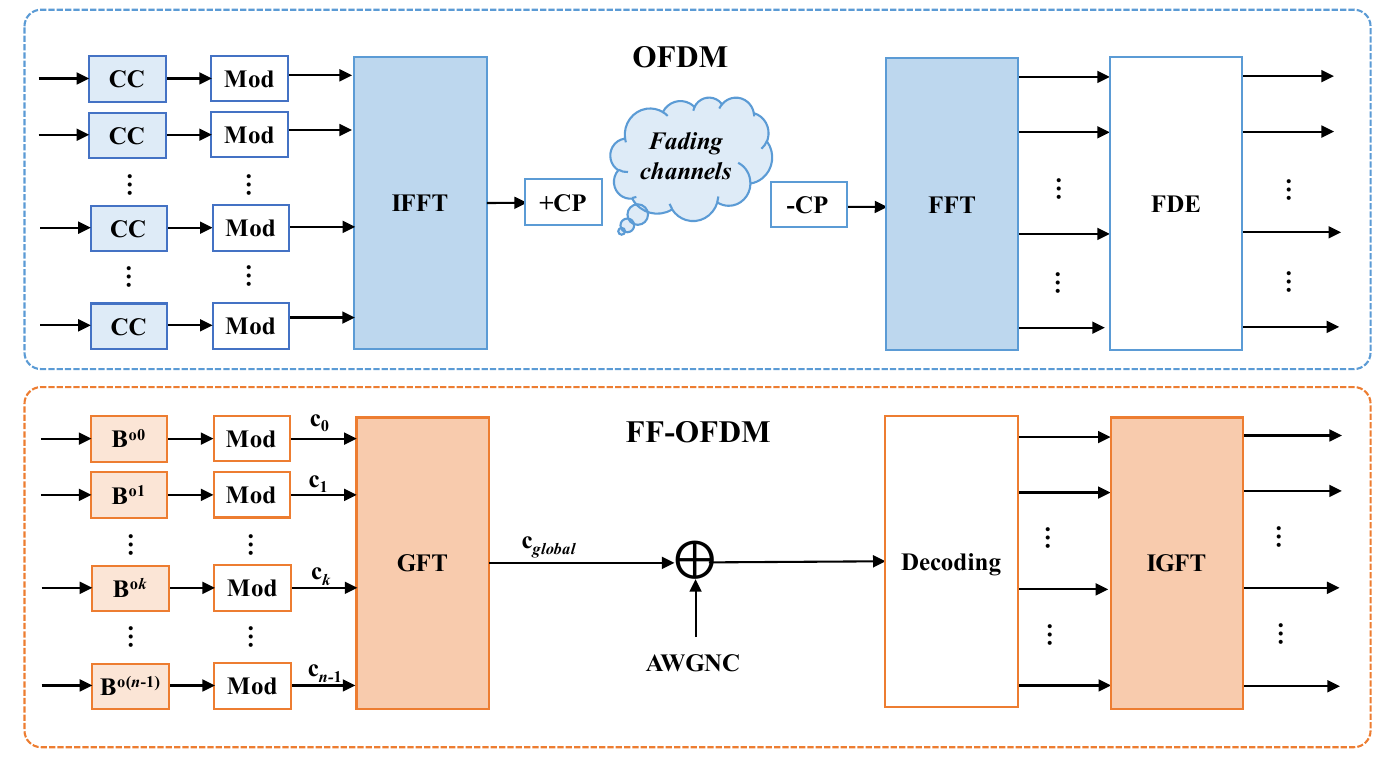}
  \caption{
  Diagram of OFDM and FF-OFDM systems. (CC: channel code; Mod: modulator;
  FFT: Fast Fourier Transform; 
  IFFT: Inverse Fast Fourier Transform;
  +CP: add cyclic prefix; -CP: remove cyclic prefix; FDE: frequency-domain equalization; GFT: Galois Fourier Transform; IGFT: Inverse Galois Fourier Transform; AWGNC: Additive White Gaussian Noise Channel.)}
 \label{fig:SystemMod}
\end{figure*}

\section{System Model of the FF-OFDM System} \label{sect3}

In this section, we introduce the system model of the proposed FF-OFDM system. To provide a clear context for our architecture, it is instructive to first briefly review the classical OFDM transmission paradigm.

In a standard OFDM transmitter, a high-speed data stream is divided into multiple parallel lower-rate streams via a serial-to-parallel (S/P) converter. These parallel streams are mapped to complex constellation symbols and independently modulated onto a set of orthogonal subcarriers. In practice, this orthogonal multiplexing is efficiently realized in the digital domain using the Inverse Fast Fourier Transform (IFFT), which synthesizes the frequency-domain symbols into a composite time-domain signal.

The block diagrams of both the classical OFDM system and the proposed FF-OFDM system are contrasted in Fig. \ref{fig:SystemMod}. Drawing a profound architectural parallel to the classical model, the proposed FF-OFDM framework substitutes the continuous complex-domain operations with rigorous algebraic finite-field counterparts. Rather than utilizing physical frequency subcarriers and the IFFT, our scheme employs orthogonal-like algebraic cyclic codes and the Galois Fourier Transform (GFT).

Suppose there are $ns$ distinct data streams to be transmitted in the FF-OFDM system. The complete transmission and joint decoding processes are detailed in the following subsections.

\subsection{The FF-OFDM Transmitter}
The proposed FF-OFDM transmission scheme is designed to multiplex $ns$ independent binary data streams into a single structured signal. The global encoding process mirrors the logic of classical OFDM: local encoding acts as subcarrier modulation, serial-to-parallel (S/P) conversion acts as the interleaver, and the GFT acts as the orthogonal synthesizer. The complete transmission process is detailed in the following steps.

\subsubsection{Step 1: Local Encoding and Subcarrier Modulation}
The $ns$ independent binary data streams are partitioned into $n$ distinct groups, denoted as $Q_0, Q_1, \ldots, Q_{n - 1}$. Each group $Q_k$ ($0 \leq k < n$) is responsible for processing $s$ binary streams. 
For a given group $Q_k$ where $1\le k < n$, the $s$ binary data streams are locally encoded into $s$ binary codewords ${\bf u}_{k,0}, {\bf u}_{k,1}, ..., {\bf u}_{k,s-1}$ using the $k$-th Hadamard equivalent code $C^{\circ k}$. The $0$-th group $Q_0$ encodes its streams using the $(n,n-1)$ SPC code $C_0$. Each binary codeword has a length of $n$.

To modulate these $s$ parallel binary streams onto their assigned algebraic subcarrier $k$, they are linearly composed into a single \textit{composite codeword} over $\mathrm{GF}(2^s)$. Let $\alpha$ be a primitive element in $\mathrm{GF}(2^s)$ such that $\{1, \alpha, \ldots, \alpha^{s -1}\}$ forms a basis. The composite codeword for the $k$-th group is formulated as:
\begin{equation}
\mathbf{c}_k = \mathbf{u}_{k,0} + \mathbf{u}_{k,1}\alpha + \cdots + \mathbf{u}_{k,s-1}\alpha^{s - 1},
\end{equation}
which can be explicitly written as an $n$-tuple vector:
\begin{equation*}
\mathbf{c}_k = (c_{k,0}, c_{k,1}, \ldots, c_{k,j}, \ldots, c_{k,n-1}), \quad c_{k,j} \in \mathrm{GF}(2^s).
\end{equation*}
By the linearity of the code's null space, it follows that $\mathbf{c}_k ({\bf B}^{\circ k})^{\rm T} = \mathbf{0}$ over $\mathrm{GF}(2^s)$.

\subsubsection{Step 2: Serial-to-Parallel (S/P) Vector Extraction}
In classical OFDM, continuous data is converted into parallel blocks prior to the IFFT. In the FF-OFDM framework, this is achieved by synchronously extracting the $j$-th symbol from all $n$ composite codewords (${\bf c}_0, {\bf c}_1, \ldots, {\bf c}_{n-1}$).

For a specific indexing position $j$ ($0 \leq j < n$), the $n$ symbols are extracted to form a $1 \times n$ parallel input vector over $\mathrm{GF}(2^s)$:
\begin{equation}
\mathbf{c}[j] = (c_{0,j}, c_{1,j}, \ldots, c_{n-1,j}).
\end{equation}
This S/P extraction perfectly emulates the behavior of a block interleaver. By distributing the adjacent symbols of any single codeword $\mathbf{c}_k$ across $n$ different parallel vectors, the system inherently acquires maximum diversity against burst errors.

\subsubsection{Step 3: GFT Synthesis and Serialization}
The multiplexing is finalized by applying the GFT to each parallel vector $\mathbf{c}[j]$. Analogous to the IFFT in the complex domain, the GFT maps the subcarrier-domain symbols into the transmission-domain using the $n \times n$ Vandermonde matrix $\bf V$, given by:
\begin{equation}
\mathbf{c}^{(\mathcal{F})}[j] = \mathbf{c}[j] \cdot {\bf V}, \quad 0 \leq j < n.
\end{equation}
Here, $\mathbf{c}^{(\mathcal{F})}[j]$ is a $1 \times n$ vector over $\mathrm{GF}(2^s)$ representing the $j$-th segment of the globally coupled signal.

The $n$ transformed vectors are then serialized to construct the final global codeword $\mathbf{c}_{global}$:
\begin{equation}
\mathbf{c}_{global} = \left(\mathbf{c}^{(\mathcal{F})}[0], \mathbf{c}^{(\mathcal{F})}[1], \ldots, \mathbf{c}^{(\mathcal{F})}[n-1]\right).
\end{equation}
Note that $\mathbf{c}_{global}$ is a vector of length $n^2$ over $\mathrm{GF}(2^s)$. To prepare for physical transmission, each $\mathrm{GF}(2^s)$ symbol is decomposed into its $s$ binary constituent bits, yielding a binary sequence of length $sn^2$. Finally, these $sn^2$ bits are mapped to physical modulation symbols (e.g., via BPSK, where $0 \rightarrow +1$ and $1 \rightarrow -1$) to generate the transmit signal vector $\mathbf{x} \in \{-1,+1\}^{sn^2}$.

\vspace{-0.2in}

\subsection{The FF-OFDM Receiver}
The receiver architecture demonstrates the most profound advantage of the FF-OFDM paradigm. In classical OFDM, the FFT (demultiplexing) is performed \textit{before} channel decoding. However, in our proposed architecture, the GFT operation mathematically weaves the separate subcarriers into a unified global QC-LDPC structure. Therefore, the receiver performs joint iterative decoding \textit{first}, followed by the IGFT as shown in Fig. \ref{fig:SystemMod}.

Assume the modulated signal vector $\mathbf{x}$ is transmitted over an Additive White Gaussian Noise Channel (AWGNC). The received signal vector $\mathbf{y} \in \mathbb{R}^{sn^2}$ is modeled as:
\begin{equation}
\mathbf{y} = \mathbf{x} + \mathbf{z},
\end{equation}
where $\mathbf{z}$ represents the AWGNC vector whose components are independent and identically distributed (i.i.d.) Gaussian random variables with zero mean and variance $\sigma^2$. It is important to emphasize that while classical OFDM is primarily designed to combat multipath fading, the FF-OFDM system evaluated herein focuses on the AWGNC. This rigorously isolates and validates the pure coding gain and reliability improvement achieved through algebraic finite-field multiplexing.

\subsubsection{Step 1: Soft-Information Extraction and Joint Parallel Decoding}
Unlike hard-decision demodulation, the receiver directly utilizes the unquantized soft information from the received signal $\mathbf{y}$. Assuming a BPSK mapping ($0 \rightarrow +1, 1 \rightarrow -1$), the receiver computes the log-likelihood ratio (LLR) sequence for the $sn^2$ transmitted bits. The global LLR vector, denoted as $\mathbf{L} \in \mathbb{R}^{sn^2}$, is mapped from $\mathbf{y}$ as
\(
\mathbf{L} = \frac{2\mathbf{y}}{\sigma^2}.
\)
Since each $\mathrm{GF}(2^s)$ symbol in the global codeword is formulated from $s$ binary bits, the global LLR vector $\mathbf{L}$ is naturally partitioned into $s$ parallel soft-information sub-vectors:
\begin{equation}
\mathbf{L}_l = (L_{0,l}, L_{1,l}, \ldots, L_{n^2-1,l}), \quad 0 \leq l < s.
\end{equation}
Here, each sub-vector $\mathbf{L}_l \in \mathbb{R}^{n^2}$ directly represents the soft information associated with the $l$-th binary constituent layer of the global codeword.

These $s$ parallel LLR vectors can be decoded independently and simultaneously using binary soft-decision decoders (e.g., the Min-Sum Algorithm) based on the global parity-check matrix ${\bf H}_{global}$, the algebraic structure of which will be rigorously derived in Section \ref{sect4}. During the iterative belief propagation process, the nodes corresponding to different data streams share reliability information through the edges of the global LDPC Tanner graph. Upon convergence, the hard decisions yield the $s$ corrected binary constituent vectors $\tilde{\mathbf{v}}_0, \ldots, \tilde{\mathbf{v}}_{s-1} \in \{0, 1\}^{n^2}$. These binary vectors are then recomposed using the basis $\{1, \alpha, \ldots, \alpha^{s-1}\}$ into the corrected global estimate $\tilde{\mathbf{c}}_{global}$ over $\mathrm{GF}(2^s)$.

\subsubsection{Step 2: Inverse GFT (Finite-Field FFT)}
With the errors optimally corrected, the receiver demultiplexes the signal. The corrected global codeword $\tilde{\mathbf{c}}_{global}$ is partitioned back into $n$ parallel vectors $\tilde{\mathbf{c}}^{(\mathcal{F})}[0], \ldots, \tilde{\mathbf{c}}^{(\mathcal{F})}[n-1]$. The Inverse GFT (analogous to the FFT) is applied to each vector using the inverse Vandermonde matrix ${\bf V}^{-1}$:
\begin{equation}
\tilde{\mathbf{c}}[j] = \tilde{\mathbf{c}}^{(\mathcal{F})}[j] \cdot \mathbf{V}^{-1}, \quad 0 \leq j < n.
\end{equation}
This operation perfectly recovers the $n$ S/P vectors $\tilde{\mathbf{c}}[j]$.

\subsubsection{Step 3: Parallel-to-Serial (P/S) Conversion and Message Retrieval}
The recovered parallel vectors $\tilde{\mathbf{c}}[j] = (\tilde{c}_{0,j}, \tilde{c}_{1,j}, \ldots, \tilde{c}_{n-1,j})$ are subjected to a Parallel-to-Serial (P/S) conversion (deinterleaving). By gathering the $j$-th components across all $n$ parallel vectors, the receiver perfectly reconstructs the $n$ composite codewords:
\begin{equation}
\tilde{\mathbf{c}}_k = (\tilde{c}_{k,0}, \tilde{c}_{k,1}, \ldots, \tilde{c}_{k,n-1}), \quad 0 \leq k < n.
\end{equation}
Finally, each composite codeword $\tilde{\mathbf{c}}_k$ is delivered to the $k$-th receiver group and decomposed over $\mathrm{GF}(2)$. The $s$ estimated binary codewords $\tilde{\mathbf{u}}_{k,0}, \ldots, \tilde{\mathbf{u}}_{k,s-1}$ are obtained, and the original $ns$ user data streams are successfully extracted.

\section{Theoretical Analysis of the Global Code Structure}\label{sect4}
The extraordinary error-correction capability of the joint FF-OFDM receiver hinges entirely on the algebraic properties of the globally coupled signal $\mathbf{c}_{global}$. In this section, we provide a rigorous mathematical derivation demonstrating how the physical processes of cascading, serial-to-parallel (S/P) extraction, and GFT synthesis naturally yield a structured QC-LDPC code based on partial geometries.

\subsection{Algebraic Formulation of the Multiplexing Process}
According to the FF-OFDM transmitter (Section \ref{sect3}), before the S/P extraction, the encoder mathematically cascades the $n$ composite codewords $\mathbf{c}_0, \mathbf{c}_1, \ldots, \mathbf{c}_{n-1}$ into a \textit{cascaded composite (CC) codeword} $\mathbf{c}_{casc} = (\mathbf{c}_0, \mathbf{c}_1, \ldots, \mathbf{c}_{n-1})$ over $\mathrm{GF}(2^s)$ of length $n^2$. The collection of $2^{(n - m + 1)(n - 1)}$ possible CC codewords forms a linear code $C_{casc}$ over $\mathrm{GF}(2^s)$. The parity-check matrix of this CC code is the following $mn \times n^2$ block-diagonal matrix over $\mathrm{GF}(2^s)$:
\begin{equation}
  {\bf H}_{casc} = \begin{pmatrix}
    {\bf B}^{\circ 0} & & & \\
    & {\bf B}^{\circ 1} & & \\
    & & \ddots & \\
    & & & {\bf B}^{\circ (n - 1)}
    \end{pmatrix},
\end{equation}
which is an $n \times n$ diagonal array of matrices of size $m \times n$. 
The parity-check matrices ${\bf B}^{\circ 0}$ for the SPC code $C_0$ and the parity-check matrices ${\bf B}^{\circ 1}, \ldots, {\bf B}^{\circ (n-1)}$ for the $n-1$ equivalent cyclic codes lie exactly on its main diagonal.

The S/P vector extraction physically functions as a deterministic block interleaver, which interleaves the code symbols of the $n$ constituent composite codewords in $\mathbf{c}_{casc}$. This interleaving results in an \textit{interleaved CC (ICC) codeword} over $\mathrm{GF}(2^s)$:
\[ 
\mathbf{c}_{casc}^\pi = (\mathbf{c}[0], \mathbf{c}[1], \ldots, \mathbf{c}[n-1]). 
\]
Applying this identical interleaving permutation to the columns and rows of the block-diagonal matrix ${\bf H}_{casc}$, we obtain the parity-check matrix of the ICC code, denoted as ${\bf H}_{casc}^\pi$:
\begin{equation}
  {\bf H}_{casc}^\pi = [{\bf D}_{i,j}]_{0 \leq i < m, 0 \leq j < n},
\end{equation}
where each ${\bf D}_{i,j}$ is an $n \times n$ diagonal matrix over $\mathrm{GF}(2^s)$ formed by the root $\beta^{l_i}$ of the generator polynomial ${\bf g}(X)$:
\begin{equation}
  {\bf D}_{i,j} = \text{diag}\left(1, \beta^{jl_i}, \beta^{2jl_i}, \ldots, \beta^{(n - 1)jl_i}\right).
\end{equation}
Hence, ${\bf H}_{casc}^\pi$ transforms into an $m \times n$ array of $n \times n$ diagonal matrices over $\mathrm{GF}(2^s)$.

\subsection{Parity-Check Matrix and Partial Geometry of the Global Code}
The final multiplexing step applies the GFT to each parallel vector using the Vandermonde matrix $\bf V$. In the global matrix domain, this is equivalent to transforming the interleaved sequence by a global block-diagonal matrix ${\bf V}_{blk} = \text{diag}({\bf V}, {\bf V}, \ldots, {\bf V})$. 

Consequently, the parity-check matrix for the globally coupled FF-OFDM code $\mathbf{c}_{global}$ undergoes a similarity transformation:
\begin{equation} \label{eq:H_global_derivation}
  \begin{aligned}
    {\bf H}_{global} 
    &= {\bf V}_{blk} \cdot {\bf H}_{casc}^\pi \cdot {\bf V}_{blk}^{-1} \\
    &= \text{diag}({\bf V}, \ldots, {\bf V}) \cdot {\bf H}_{casc}^\pi \cdot \text{diag}({\bf V}^{-1}, \ldots, {\bf V}^{-1}) \\
    &= [{\bf V} {\bf D}_{i,j} {\bf V}^{-1}]_{0 \leq i < m, 0 \leq j < n}.
  \end{aligned}
\end{equation} 

It is a fundamental property of the finite-field Fourier transform that multiplies a geometric-progression diagonal matrix by the Fourier basis and its inverse yields a circulant permutation matrix. Specifically, multiplying the term $\mathbf{V}\mathbf{D}_{i,j}\mathbf{V}^{-1}$ out reveals that it mathematically evaluates to the $n \times n$ CPM-dispersion of the entry $\beta^{j l_i}$. Therefore, the global parity-check matrix $\mathbf{H}_{global}$ for the global $C_{global}$ for the FF-OFDM is exactly the CPM-dispersion of the original parity-check matrix ${\bf B}$ of the base cyclic code $C$:
\begin{equation}
  \mathbf{H}_{global} = \text{CPM}(\mathbf{B}).
\end{equation}
This derivation provides a profound insight: the physical OFDM-like multiplexing over a finite field mathematically translates to generating the matrix $\CPM({\bf B})$. Since ${\bf H}_{global}$ is an $m \times n$ array of $n \times n$ circulant submatrices, the null space of ${\bf H}_{global}$ over GF($2^s$) gives a quasi-cyclic (QC) code $C_{global}$ over $\mathrm{GF}(2^s)$.

Furthermore, as established in Section \ref{sect2}, $\CPM({\bf B})$ is the line-point incidence matrix of a subgeometry of the partial geometry $\PaG(n, n, n - 1)$. Hence, the global encoding for the FF-OFDM system links the global code to a partial geometry. Since ${\bf H}_{global}$ is a submatrix of the line-point incidence matrix of the partial geometry $\PaG(n, n, n - 1)$, it satisfies the RC-constraint and has orthogonal structure.

For $n > 7$, ${\bf H}_{global}$ is a low-density parity-check matrix over $\mathrm{GF}(2)$ and the global code $C_{global}$ is a QC-LDPC code over $\mathrm{GF}(2^s)$ whose associated Tanner graph has girth at least 6. Hence, the global code for the multi-user communication system can be decoded iteratively with a soft-decision iterative algorithm based on an LDPC matrix over $\mathrm{GF}(2)$. 

\vspace{-0.2in}

\subsection{The Binary Decomposition Theorem}
While $\mathbf{c}_{global}$ is constructed over $\mathrm{GF}(2^s)$, deploying a nonbinary LDPC decoder directly involves excessive computational complexity. However, the structural orthogonality of the FF-OFDM global code enables exact binary decomposition, governed by the following theorem \cite{Ref_11, Ref_19}:

\begin{theorem} \label{thm1}
A vector $\bf{v}$ of length $n^2$ over $\mathrm{GF}(2^s)$ is a codeword in $C_{global}$ if and only if each of its binary constituent vectors ${\bf v}_{b,l}$ ($0 \leq l < s$) is a codeword in the binary null space of $\mathbf{H}_{global}$. That is, ${\bf v} {\bf H}_{global}^{\rm T} = \bf{0}$ over $\mathrm{GF}(2^s)$ if and only if ${\bf v}_{b,l} {\bf H}_{global}^{\rm T} = \bf{0}$ over $\mathrm{GF}(2)$ for all $0 \leq l < s$.
\end{theorem}

\begin{proof}

 Let ${\bf v} = (v_0, v_1, \ldots, v_{n^2-1})$ be a vector of $n^2$ components over $\mathrm{GF}(2^s)$. Let $\alpha$ be a primitive element in $\mathrm{GF}(2^s)$ such that $\{\alpha^0=1, \alpha, \ldots, \alpha^{s-1}\}$ forms a basis of $\mathrm{GF}(2^s)$ over $\mathrm{GF}(2)$.

 For $0 \leq j < n^2$, we express the $j$-th component $v_j$ as a linear combination of the basis elements:
\[ 
v_j = v_{j,0} + v_{j,1}\alpha + \cdots + v_{j,s - 1}\alpha^{s - 1}, 
\]
where $v_{j,l} \in \mathrm{GF}(2)$ for $0 \leq l < s$. We can thus decompose the entire vector $\bf{v}$ into $s$ parallel binary vectors over $\mathrm{GF}(2)$:
\[ 
{\bf v}_{b,l} = (v_{0,l}, v_{1,l}, \ldots, v_{n^2-1,l}), \quad 0 \leq l < s.
\]
The vector $\bf{v}$ is simply the $2^s$-ary composition of these binary constituent vectors:
\[ 
  {\bf v} = {\bf v}_{b,0} + {\bf v}_{b,1}\alpha + \cdots + {\bf v}_{b,s-1}\alpha^{s - 1}. 
\]

Evaluating the syndrome of the global vector by right-multiplying by the transposed global parity-check matrix, we have:
\begin{equation}
  \begin{aligned}
  {\bf v} {\bf H}_{global}^{\rm T} 
  &= ({\bf v}_{b,0} + {\bf v}_{b,1}\alpha + \cdots + {\bf v}_{b,s-1}\alpha^{s - 1}) {\bf H}_{global}^{\rm T}\\
  &= \sum_{l=0}^{s-1} \left({\bf v}_{b,l} {\bf H}_{global}^{\rm T}\right) \alpha^{l}.
  \end{aligned}
\end{equation}

Because ${\bf H}_{global}$ is a binary matrix (i.e., its entries belong to $\mathrm{GF}(2)$), each product term $\left({\bf v}_{b,l} {\bf H}_{global}^{\rm T}\right)$ inherently yields a vector over $\mathrm{GF}(2)$. Since the elements $\{1, \alpha, \ldots, \alpha^{s-1}\}$ are strictly linearly independent over $\mathrm{GF}(2)$, the summation $\sum_{l=0}^{s-1} \left({\bf v}_{b,l} {\bf H}_{global}^{\rm T}\right) \alpha^{l} = {\bf 0}$ holds true \textit{if and only if} every individual coefficient evaluates to zero.

Therefore, ${\bf v} {\bf H}_{global}^{\rm T} = {\bf 0}$ if and only if ${\bf v}_{b,l} {\bf H}_{global}^{\rm T} = {\bf 0}$ for all $0 \leq l < s$. This completes the proof.
\end{proof}

This theorem provides the rigorous mathematical foundation for Step 1 of the FF-OFDM receiver (Section \ref{sect3}). It validates that instead of utilizing an impractically complex nonbinary decoder, the receiver can optimally decode the $ns$ multiplexed streams simultaneously using $s$ parallel binary soft-decision QC-LDPC decoders, unlocking massive multi-user connectivity with highly manageable hardware efficiency.

\section{Performance Metrics and Decoding Complexity Analysis} \label{sect5}
In this section, we define the error performance metrics used to evaluate the proposed FF-OFDM system and provide a rigorous analysis of its computational decoding complexity.

\subsection{Measure of Error Performance}
To accurately assess the reliability of the globally coded multiuser communication system, we utilize two primary performance metrics: the Global Error Rate (GER), denoted by $P_{GER}$, and the Word Error Rate (WER), denoted by $P_{WER}$. 

The GER, $P_{GER}$, is defined as the probability that a received global FF-OFDM codeword is decoded incorrectly (i.e., at least one bit in the global block fails to converge to the correct value). The WER, $P_{WER}$, is defined as the probability that a specific composite codeword (corresponding to a subcarrier data block) extracted from the decoded global codeword is incorrect. 

The error probability $P_{WER}$ of a composite codeword is analytically bounded by the global error probability $P_{GER}$. Let $\lambda$ be the average number of composite codewords in a decoded global codeword that are incorrectly decoded. The relationship is given by:
\begin{equation}\label{eq:wer_ger_relation}
  P_{WER} = \frac{\lambda}{n} \cdot P_{GER}.
\end{equation}
Since by definition $1 \leq \lambda \leq n$, the global error probability $P_{GER}$ serves as a strict upper bound on the subcarrier word error probability $P_{WER}$. In our extensive simulations, we observe that the gap between these two error probabilities is remarkably tight. This reflects the holistic nature of the global LDPC graph: if a received global codeword is not decoded correctly, the residual errors are typically distributed such that most of the $n$ constituent codewords are affected. 

To benchmark the coding gain achieved by the FF-OFDM multiplexing, we establish a baseline error probability, denoted as $P^{*}_{WER}$. This baseline represents the probability of decoding failure when each $(n, n-m)$ cyclic codeword is transmitted and decoded entirely independently using standard soft-decision or hard-decision algorithms, without any global coupling. The true advantage of the FF-OFDM system is quantified by comparing $P_{WER}$ against $P^{*}_{WER}$.

\vspace{-0.2in}
\subsection{Decoding Complexity and Hardware Parallelism}
A critical metric for practical communication systems is the computational complexity of the receiver. Conventionally, joint decoding of a codeword over a high-order Galois field $\mathrm{GF}(2^s)$ requires nonbinary LDPC decoding algorithms (e.g., the Q-ary Sum-Product Algorithm, QSPA). The complexity of such algorithms is typically scaled by $\mathcal{O}(q \log q)$ or worse, where $q = 2^s$. For large $s$, this exponential growth yields prohibitive hardware complexity and latency.

The proposed FF-OFDM scheme elegantly circumvents this bottleneck through the Binary Decomposition Theorem (Theorem \ref{thm1}). Instead of a single nonbinary decoder, the receiver employs $s$ parallel binary decoders.

Suppose we decode the global code using a scaled Min-Sum Algorithm (MSA) \cite{Ref_11, Ref_19}. 
The global parity-check matrix $\mathbf{H}_{global} = \CPM({\bf B})$ is a $mn \times n^2$ matrix with column and row weights $m$ and $n$, respectively. The number of 1-entries (which correspond to the edges in the Tanner graph) is $mn^2$.

During each iteration of the MSA, messages are passed along these edges between variable nodes and check nodes. For a single binary layer, updating the reliabilities requires a small, constant number of real-number operations (additions and comparisons) per edge, i.e., typically around 3 operations per edge in a highly optimized MSA implementation. Because the FF-OFDM receiver processes $s$ independent binary layers, the total number of edges processed per iteration is $smn^2$. Therefore, the average number of real-number computations required per iteration for the entire received global codeword is:
\begin{equation} \label{eq:complexity_global}
  N_{avg} \approx 3smn^2.
\end{equation}

To understand the efficiency on a subcarrier basis, we amortize this total complexity over the $ns$ constituent user streams. The average number of computations required to update the reliabilities for a single received user codeword of length $n$ per iteration is exactly:
\begin{equation} \label{eq:complexity_user}
  \frac{N_{avg}}{ns} = \frac{3smn^2}{ns} = 3mn.
\end{equation}

This is a profoundly significant analytical result. It mathematically demonstrates that the FF-OFDM scheme achieves the massive coding gains and diversity of an ultra-long global LDPC block code, yet the amortized decoding complexity scales merely linearly with the base code parameters ($m$ and $n$).

Furthermore, from a hardware architecture perspective, the $s$ binary QC-LDPC decoders can operate strictly in parallel without any data dependency between them during the iteration process. This parallelism ensures that the decoding latency of the FF-OFDM system is essentially identical to the latency of decoding a single binary LDPC code. Consequently, the proposed architecture provides an exceptional trade-off, delivering near-capacity error performance while maintaining highly manageable silicon area and energy consumption.

\section{Simulation Examples and Performance Analysis} \label{sect6}

To comprehensively evaluate the proposed FF-OFDM system, extensive Monte Carlo simulations are conducted over the AWGNC using BPSK signaling. In all scenarios, the receiver employs the joint binary parallel decoding scheme detailed in Section \ref{sect3}, utilizing the MSA with a designated scaling factor to optimize belief propagation on the global QC-LDPC graph.

\subsection{Scenario One: High-Capacity Point-to-Point Multiplexing via Binary Cyclic Base Codes}
In this scenario, we evaluate the FF-OFDM system configured for aggregated point-to-point transmission. The premise is that a centralized transmitting node (e.g., a baseband processing unit) needs to multiplex $ns$ independent binary data streams into a single high-capacity waveform. These streams are locally encoded using binary cyclic codes, aggregated into composite $\mathrm{GF}(2^s)$ subcarrier blocks, and multiplexed via the GFT for transmission over a single AWGNC.

\begin{figure}[tbp]
  \centering
  % Please replace with your actual EPS file
  \includegraphics[height=0.4\textwidth, width=0.49\textwidth]{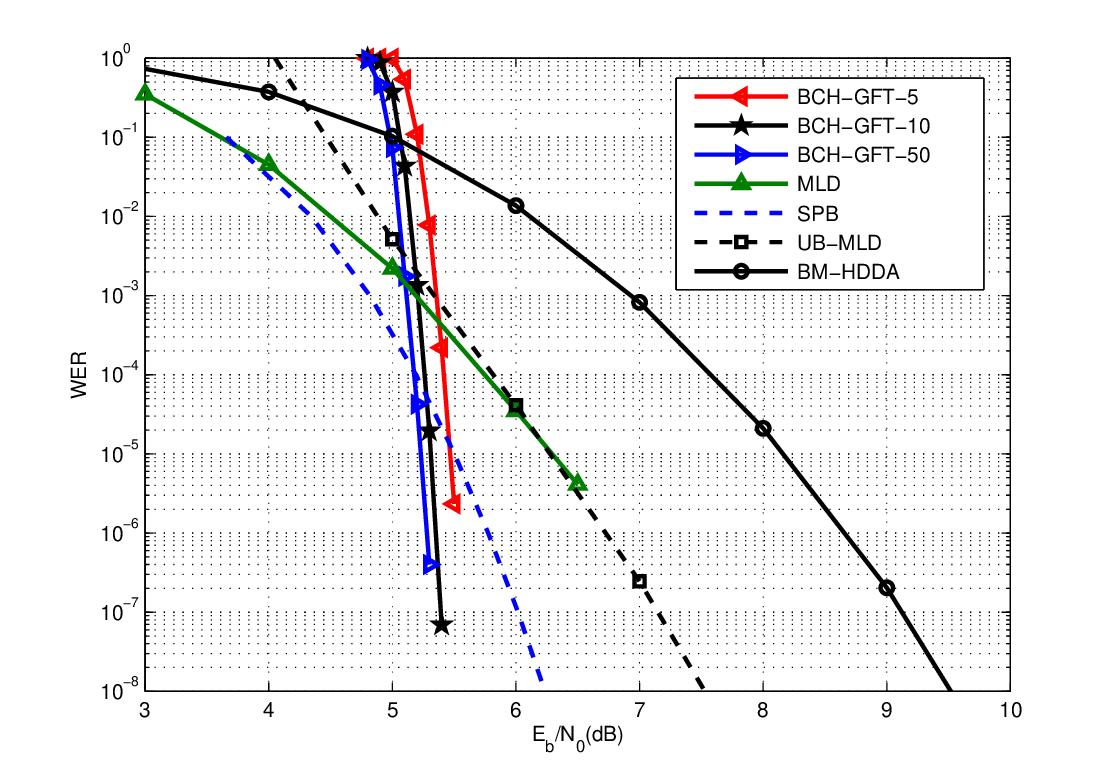}
  \caption{WER performance of the point-to-point FF-OFDM multiplexed system utilizing a (127, 113) binary BCH base code ($n=127, s=7$).}
 \label{fig:eg1}
\end{figure}

\begin{example} \label{eg1}
In this example, we consider an FF-OFDM system utilizing a BCH code and its Hadamard equivalents as the orthogonal algebraic subcarriers. Let $s = 7$ and $\beta$ be a primitive element in $\mathrm{GF}(2^7)$ with an order of $n = 127$. We select a $(127, 113)$ binary BCH base code $C$ with a rate of $0.89$ and a minimum distance of $5$. Its generator polynomial ${\bf g}(X)$ contains $m=14$ roots.

The parity-check matrix ${\bf B}$ of $C$ is a $14 \times 127$ matrix over $\mathrm{GF}(2^7)$. Following the theoretical derivation in Section \ref{sect4}, its CPM-dispersion $\CPM({\bf B})$ expands into a $1778 \times 16129$ binary low-density matrix. 
{It is the line-point incidence matrix} of a subgeometry of $\PaG(127, 127, 126)$, with column and row weights of $14$ and $127$, respectively.

The centralized transmitter is configured to aggregate $ns = 127 \times 7 = 889$ distinct data streams. These streams are partitioned into $127$ groups, i.e.,$Q_0, \ldots, Q_{126}$, each processing $7$ streams. For $1 \leq k \leq 126$, group $Q_k$ encodes its streams using the $k$-th Hadamard equivalent $C^{\circ k}$, while $Q_0$ uses the $(127, 126)$ SPC code. After local encoding, the $7$ binary codewords in each group are combined into a composite codeword over $\mathrm{GF}(2^7)$. The $127$ composite codewords are subjected to S/P extraction and GFT synthesis. This global coupling yields a $(16129, 14364)$ global code $C_{global}$ over $\mathrm{GF}(2^7)$, maintaining a high system transmission rate of $0.8905$.

At the receiving node, the global code $C_{global}$ is decoded using the joint parallel MSA with a scaling factor of $0.625$. The WER performance using $5$, $10$, and $50$ iterations is shown in Fig. \ref{fig:eg1}. The decoding converges extremely fast, with the gap between $10$ and $50$ iterations being less than $0.1$ dB. 

Also included in Fig. \ref{fig:eg1} is the error performance of the baseline $(127, 113)$ BCH code decoded independently with Maximum Likelihood Decoding (MLD), its Union Bound (UB), and the Sphere Packing Bound (SPB). At a WER of $10^{-5}$, the FF-OFDM global decoding with $50$ iterations achieves more than a $1$ dB coding gain over the independent MLD. Notably, for SNRs above $5.5$ dB, the waterfall curves of the global decoding strictly surpass the theoretical SPB of the independently decoded BCH code. At a WER of $10^{-7}$, the global decoding achieves a substantial $3.7$ dB coding gain over standard Berlekamp-Massey hard-decision decoding (BM-HDDA).
\end{example}

The remarkable results in Example \ref{eg1} demonstrate that by employing FF-OFDM multiplexing, even a short BCH code of minimum distance 5 can leverage massive block-length diversity to support highly reliable aggregated communications.

\begin{example} \label{eg2}
In this example, we consider an FF-OFDM system utilizing another BCH code and its Hadamard equivalents as the orthogonal algebraic subcarriers. Let $s = 7$  and $\beta$ be a primitive element in GF($2^7$) with an order of $n =127$. The chosen base BCH code $C$ is the $(127, 120)$ code with rate $0.9449$ and minimum distance of $3$. It is a single error-correcting code. In fact, it is a Hamming code. Its generator polynomial ${\bf g}(X)$ has $\beta$ and its $6$ conjugates $\beta^2, \beta^4, \beta^8, \beta^{16}, \beta^{32}$, and $\beta^{64}$ as roots.

The parity-check matrix $\bf B$ of $C$ is a $7 \times 127$ matrix over GF($2^7$). The CPM-dispersion CPM(${\bf B}$) of $\bf B$ with respect to the cyclic group of GF($2^7$) is a $7 \times 127$ array ${\bf H}_{global}$ of CPMs of size $127 \times 127$. The array ${\bf H}_{global}$ is an $889 \times 16129$ low-density matrix with column and row weights $7$ and $127$, respectively. It is the line-point incidence matrix of a subgeometry of the partial geometry PaG$(127, 127, 126)$ as presented in Example 1.

The centralized transmitter is configured to aggregate $ns =127\times7=889$  distinct data streams. These streams are partitioned into $127$ groups, i.e, $Q_0, Q_1, ..., Q_{126}$, each processing $7$ streams. For $1 \le k \le 126$ , group $Q_k$ encodes its streams using the $k$-th Hadamard equivalent $C^{\circ k}$, while $Q_0$ uses the $(127,126)$ SPC code. After local encoding, the $7$ binary codewords in each group are combined into a composite codeword over GF($2^7$). The $127$ composite codewords are subjected to S/P extraction and GFT synthesis. This global coupling yields a $(16129, 15246)$ global LDPC code $C_{global}$ over GF($2^7$) which provides a very high system transmission rate $0.94525$.

At the receiving node, the global code $C_{global}$ is decoded using the joint parallel MSA with a scaling factor of $0.625$. The WER performances using $5$, $10$, and $50$ iterations are shown in Fig. \ref{fig:eg2}. The decoding converges extremely fast, with the gap between $10$ and $50$ iterations less than $0.1$ dB.

Also included in Fig. \ref{fig:eg2} is the error performance of the baseline $(127,120)$ BCH code decoded independently with MLD, its UB, and the SPB. At a WER of $10^{-8}$, the FF-OFDM global decoding with only $10$ iterations achieves a $3.5$ dB over the independent MLD, and $3.4$ dB and $1.5$ dB coding gains over the UB and SPB on MLD. At a WER of $10^{-8}$, the global decoding achieves a substantial $5.2$ dB coding gain over the BM-HDDA.

In this example, we demonstrate that even using a simple binary cyclic code as the base code, globally coupled FF-OFDM can achieve highly reliable information transmission  for a high-capacity point-to-point multiplexing communication system.
\end{example}

\begin{figure}[tbp]
  \centering
  \includegraphics[height=0.4\textwidth, width=0.49\textwidth]{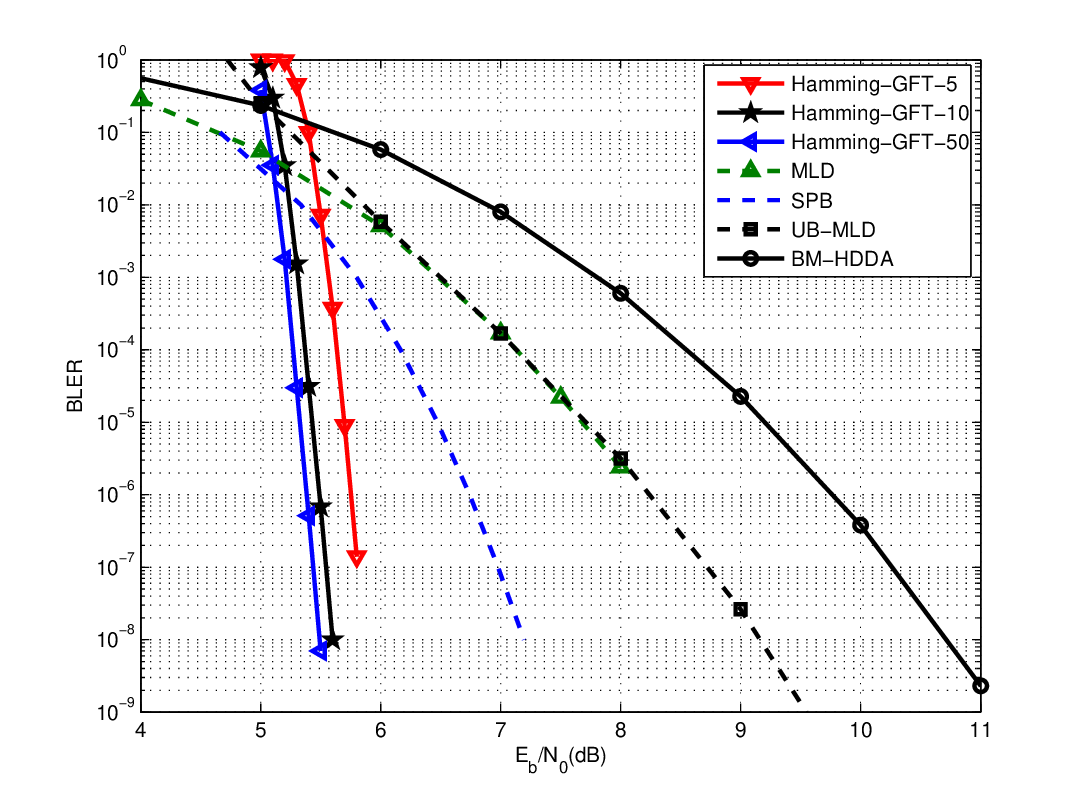}
  \caption{WER performance of the point-to-point FF-OFDM multiplexed system utilizing a $(127, 120)$ binary BCH base code over $(n=127, s=7)$.}
 \label{fig:eg2}
\end{figure}

\vspace{-0.2in}
\subsection{Scenario Two: Point-to-Point Multiplexing via Nonbinary Cyclic Base Codes} 
In this scenario, we elevate the system by employing nonbinary cyclic codes, specifically Reed-Solomon (RS) codes, as the underlying base codes. The physical premise remains a point-to-point aggregated transmission, but to maximize the nonbinary distance properties, the $ns$ binary data streams are first aggregated and mapped into $\mathrm{GF}(2^s)$ symbols \textit{prior} to local encoding. This ensures the algebraic subcarriers operate directly on nonbinary field elements before the GFT synthesis.

\vspace{-0.1in}
\begin{example} \label{eg3}
Let $s = 7$ and $n = 127$. We select a $(127, 121)$ RS code over $\mathrm{GF}(2^7)$ with a rate of $0.952$ and a minimum distance of $7$ ($m=6$). The generator polynomial has $\beta, \beta^2, \ldots, \beta^6$ as roots. The CPM-dispersion $\CPM(\mathbf{B})$ is a $762 \times 16129$ binary low-density matrix serving as a subgeometry of $\PaG(127, 127, 126)$.

In the global encoding for this $889$-stream system, $7$ parallel binary streams (each providing $121$ information bits) within group $k$ ($1 \leq k < 126$) are combined into a nonbinary message $\mathbf{b}_k$ comprising $121$ symbols over $\mathrm{GF}(2^7)$. This message is then directly encoded into the $k$-th Hadamard equivalent RS code ${C}^{\circ k}$. Group $0$ encodes into the $(127, 126)$ nonbinary SPC code. The S/P extraction and GFT synthesis globally couple these $127$ subcarriers into a $(16129, 15372)$ global QC-LDPC code ${C}_{global}$ over $\mathrm{GF}(2^7)$ with a rate of $0.9530$.

The WER performance of the global decoding using $5$, $10$, and $50$ iterations of the scaled MSA (factor $0.625$) is shown in Fig. \ref{fig:eg3}. The error performance curves drop precipitously. At a WER of $10^{-6}$, the global RS-based code decoded with $50$ MSA iterations performs within $0.3$ dB of the independent SPB. As the SNR increases, the WER curve tightly hugs the SPB trajectory, proving that the GFT-based coupling excellently preserves the Maximum Distance Separable (MDS) properties of the RS codes across the multiplexed subcarriers.
\end{example}

\begin{figure}[tbp]
  \centering
  \includegraphics[width=0.48\textwidth]{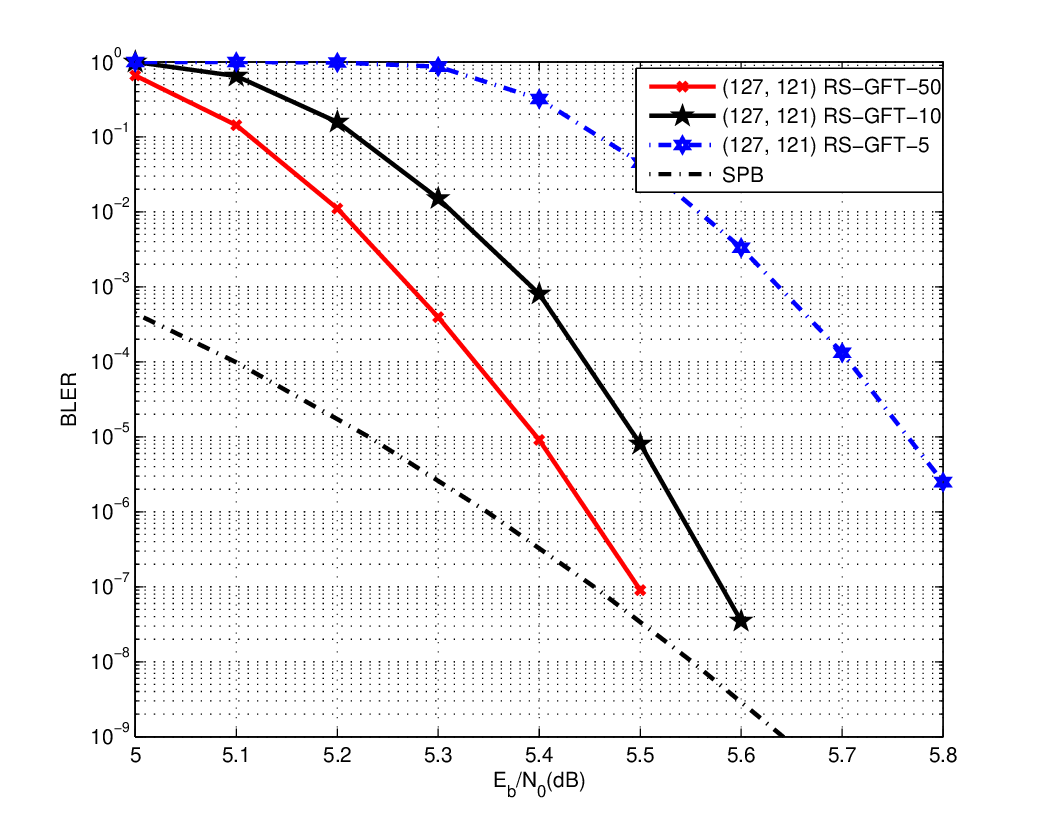}
   \caption{WER performance of the globally coded FF-OFDM point-to-point system utilizing a $(127, 121)$ nonbinary RS base code over GF$(2^7)$.}
 \label{fig:eg3}
\end{figure}

\begin{example} \label{eg4}
A critical concern for practical LDPC codes is the presence of an error floor caused by local topological trapping sets. The global code \( C_{global} \) inherently possesses an associated binary QC-LDPC code \( C_{binary} \), derived from the null space over \( \mathrm{GF}(2) \) of the incidence matrix \( \CPM({\bf B}) \). 

To physically verify the graph's integrity, we simulate the underlying \( (16129, 15372) \) binary QC-LDPC code generated in Example \ref{eg3}. As shown in Fig. \ref{fig:eg4}, the code successfully achieves a BER of \( 10^{-15} \) and a WER of nearly \( 10^{-12} \) without any visible error floors. The performance gap between 5 and 50 iterations at a BER of \( 10^{-12} \) is merely $0.5$ dB, and the gap between the performance curves of $10$ and $50$ iterations is about $0.1$ dB.
This empirical data flawlessly corroborates the theoretical partial geometry analysis: the strict RC-constraint of $\CPM(\mathbf{B})$ enforces a Tanner graph girth of at least 6, thereby eliminating length-4 cycles and ensuring absolute reliability.
\end{example}

\begin{figure}[tbp]
  \centering
  \includegraphics[height=0.4\textwidth, width=0.49\textwidth]{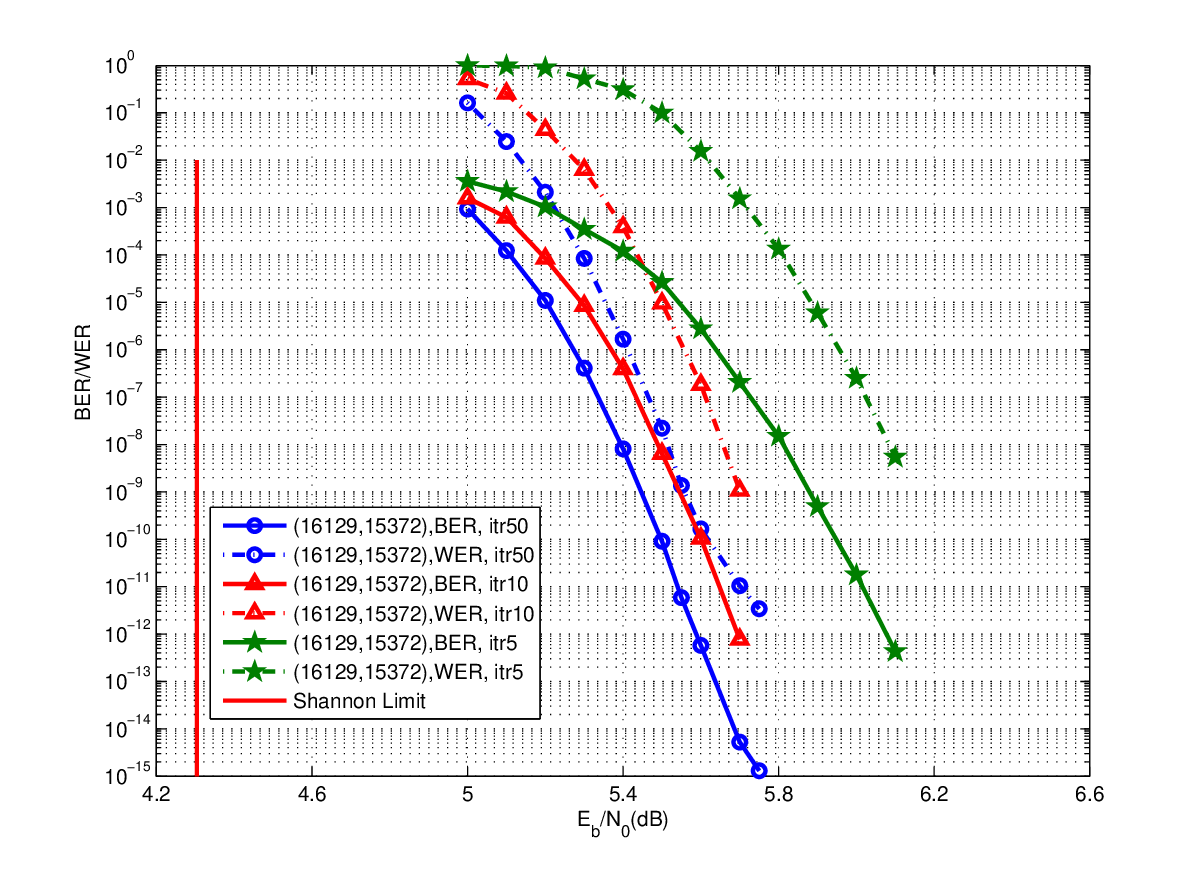}
   \caption{BER and WER performance of the underlying $(16129, 15372)$ binary QC-LDPC code from the RS-based system, demonstrating the complete absence of error floors.}
 \label{fig:eg4}
\end{figure}

\subsection{Scenario Three: FF-OFDM for Broadcast Communication Systems}
While Scenarios One and Two focus on point-to-point aggregated transmission, the FF-OFDM architecture inherently adapts to point-to-multipoint Broadcast Systems (Downlink). In this topology, a single transmitting center broadcasts a globally coupled signal to $n$ spatially distributed receiver-groups.

In transmission, the broadcaster continuously segments the incoming data into blocks of $(n - 1)s$ messages. As soon as $s$ consecutive messages are encoded into subcarrier codewords, they are combined, processed through the S/P interleaver, and shifted into the GFT-unit block by block. The synthesized global codeword is then broadcasted over the air. At the receiving end, each individual receiver captures the broadcasted waveform and independently experiences AWGN. Crucially, each receiver decodes the \textit{entire} global codeword, harvesting the massive cooperative coding gain provided by the global QC-LDPC graph, but post-processing only decomposes and extracts the specific data streams intended for its own group.

\begin{example} \label{eg5}
Let $s =11$, $n = 89$, and $m = 4$ where $89$ is a prime factor of $2^{11}-1 = 2047$ and $2047 = 23 \times 89$. Let $\alpha$ be a primitive element of the field GF($2^{11}$) and $\beta = \alpha^{23}$. The order of $\beta$ is $89$. The elements $\beta^0 = 1, \beta, ..., \beta^{88}$ form the cyclic subgroup $\bf G$ of GF($2^{11}$).

We configure the transmitter using the $(89,85)$ RS code of rate $0.9550$ over GF($2^{11}$) as the base code whose generator polynomial ${\bf g}(X)$ has $\beta, \beta^2, \beta^3$ and $\beta^4$ as roots. The parity-check matrix $\bf B$ of the code is a $4 \times 89$ matrix over GF($2^{11}$).

The CPM-dispersion CPM($\bf B$) with respect to ${\bf G}$ is a $4 \times 89$ array of CPMs of size $89 \times 89$ which is a $356 \times 7921$  matrix over GF($2$). It is the line-point incidence matrix of a subgeometry of the partial geometry PaG($89, 89, 88$) constructed by using the $89 \times 89$ Vandermonde matrix constructed based on the cyclic subgroup $\bf G$ of GF($2^{11}$).

During global encoding, a massive sequence of $(n - 1)(n - m + 1)s = 88 \times 86 \times 11 = 83,248$ binary information bits is continuously segmented. These bits are divided into $968$ messages (each of $85$ bits) and $11$ messages (each of $88$ bits). Through $\mathrm{GF}(2)$ to $\mathrm{GF}(2^{11})$ mapping, these are composed into $88$ nonbinary messages of length $85$ and $1$ nonbinary message of length $88$. These are encoded into $88$ RS codewords and $1$ SPC codeword, respectively. Following Hadamard permutation, interleaving, and GFT synthesis, the broadcast signal forms a $(7921, 7568)$ global QC-LDPC code over $\mathrm{GF}(2^{11})$.

Assuming independent AWGN channels to the receivers, the error performance at a representative receiver using $50$ MSA iterations (scaling factor $0.75$) is shown in Fig. \ref{fig:eg5}. The resulting waterfall curve drops precipitously. At a WER of $10^{-7}$, the performance gap to the theoretical independent SPB is a mere $0.16$ dB, and the curves are projected to completely overlap below a WER of $10^{-8}$. This compelling result verifies that FF-OFDM enables near-capacity broadcasting over algebraic finite-field channels, allowing individual receivers to decode with ultra-high reliability by leveraging the global redundancy structure.
\end{example}

\begin{figure}[tbp]
  \centering
  \includegraphics[height=0.4\textwidth, width=0.49\textwidth]{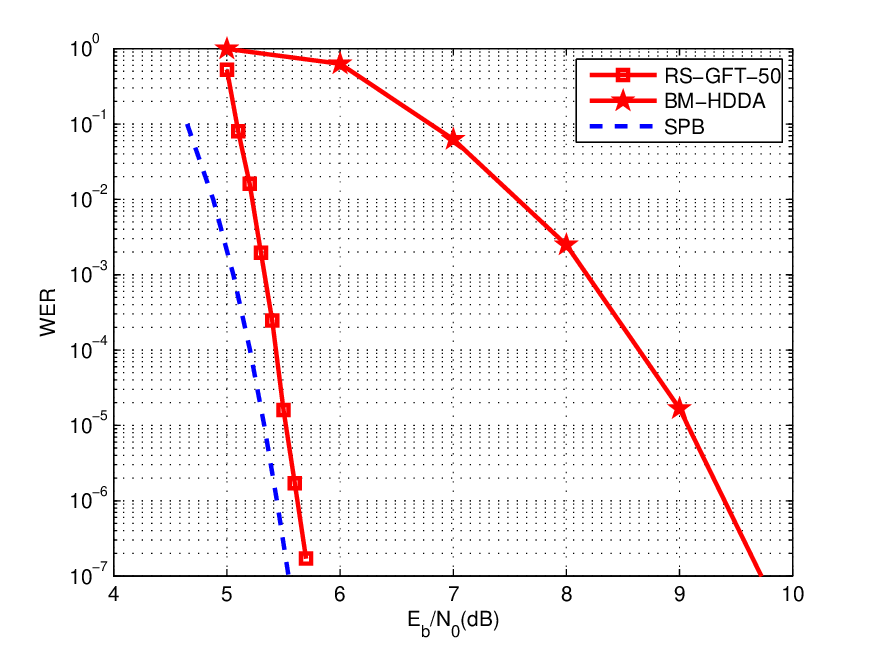}
   \caption{WER performance of the FF-OFDM broadcast communication system (Downlink) utilizing an $(89, 85)$ RS base code over $\mathrm{GF}(2^{11})$.}
 \label{fig:eg5}
\end{figure}

\section{Conclusion and Remarks}\label{sect7}
In this paper, we proposed a highly efficient global coded-multiplexing scheme, conceptualized as FF-OFDM, to provide ultra-reliable transmission for high-capacity point-to-point and broadcast communication systems. By employing a cyclic code of prime length and its Hadamard equivalents as orthogonal algebraic subcarriers, the input data streams are locally modulated, interleaved via serial-to-parallel extraction, and globally coupled through the GFT. We mathematically proved that this physical multiplexing over $\mathrm{GF}(2^s)$ intrinsically synthesizes a global QC-LDPC code, whose parity-check matrix is governed by the structural rigor of partial geometries. At the receiver, rather than demultiplexing first, the global codeword undergoes joint binary iterative soft-decision decoding. This architecture permits seamless reliability information sharing across all multiplexed subcarriers. Simulation results confirmed that the FF-OFDM receiver yields massive coding gains, closely approaching the sphere packing bound, and converges rapidly without visible error floors, all while maintaining a strictly linear amortized decoding complexity.

While the simulation examples in this study primarily utilized BCH and Reed-Solomon codes as the underlying subcarrier base codes, the FF-OFDM framework is highly versatile and generalizes to a broader algebraic context. Other classes of cyclic codes, such as cyclic Reed-Muller codes \cite{Ref_9, Ref_11}, polynomial codes \cite{Ref_22}, projective geometry codes \cite{Ref_11}, and quadratic residue codes \cite{Ref_10, Ref_11}, can be seamlessly integrated as base codes. 

Furthermore, an intriguing avenue for future research arises when the base code itself is a cyclic LDPC code. In such configurations, the FF-OFDM architecture yields a \textit{doubly iterative decodable} LDPC system. For instance, employing a $(73, 45)$ cyclic projective geometry LDPC code \cite{Ref_9, Ref_11} over $\mathrm{GF}(2^9)$ and its $71$ Hadamard equivalents can support a $657$-stream multiplexed system. In this case, the global code is a doubly iterative decodable LDPC code of length $8649$ with a minimum distance of $28$. Both the global joint detection and the local post-demultiplexing decoding can be executed via binary soft-decision algorithms. Alternatively, if the cyclic base code is relatively short, applying Ordered Statistic Decoding (OSD) \cite{Ref_9, Ref_23} or MLD at the local subcarrier level, complementing the global soft-decision receiver, could push the error performance even closer to the theoretical channel capacity.

Finally, while this foundational study rigorously validates the algebraic multiplexing principles and the theoretical coding gains of the FF-OFDM architecture over the AWGNC, practical deployment in wireless environments necessitates addressing multipath fading. In classical OFDM, time-dispersive channels are mitigated through the insertion of a cyclic prefix (CP) and frequency-domain equalization (FDE). Analogously, a critical and promising direction for future work involves extending the FF-OFDM framework to multipath fading channels. This will entail the development of finite-field equalization techniques and appropriate guard intervals to preserve the strict algebraic orthogonality of the GFT subcarriers against inter-symbol interference (ISI), thereby fully unlocking the scheme's potential for robust real-world wireless communications.

%\vspace{-0.1in}

\end{document}